\documentclass[12pt]{l4dc2021}

\usepackage{mathrsfs,mathtools,bm,bbm}
\usepackage{xcolor}

\newcommand{\E}{{\mathbb E}}

\newcommand{\Py}{{\mathbb P}}

\newcommand{\s}[2]{{#1}^{(#2)}}
\newcommand{\hp}{{\hat{p}}}
\newcommand{\dir}{{\mathscr{D}}}

\DeclareMathOperator*{\argmax}{arg\,max}

\SetKwInput{KwInput}{Input}                
\SetKwInput{KwOutput}{Output}              
\SetKw{to}{to}

\title[Bayesian UCB for Estimating Distributions]{Adaptive Sampling for Estimating Distributions: A Bayesian Upper Confidence Bound Approach}
\usepackage{times}

\author{%
 \Name{Dhruva Kartik} \Email{mokhasun@usc.edu}\\
 \Name{Neeraj Sood} \Email{nsood@usc.edu}\\
  \Name{Urbashi Mitra} \Email{ubli@usc.edu}\\
 \addr University of Southern California, Los Angeles
 \AND
 \Name{Tara Javidi} \Email{tjavidi@ucsd.edu}\\
 \addr University of California, San Diego%
}

\begin{document}

\maketitle

\begin{abstract}%
The problem of adaptive sampling for estimating probability mass functions (pmf) uniformly well is considered. Performance of the sampling strategy is measured in terms of the worst-case mean squared error. A Bayesian variant of the existing upper confidence bound (UCB) based approaches is proposed. It is shown analytically that the performance of this Bayesian variant is no worse than the existing approaches. The posterior distribution on the pmfs in the Bayesian setting allows for a tighter computation of upper confidence bounds which leads to significant performance gains in practice. Using this approach, adaptive sampling protocols are proposed for estimating SARS-CoV-2 seroprevalence in various groups such as location and ethnicity. The effectiveness of this strategy is discussed using data obtained from a seroprevalence survey in Los Angeles county.
\end{abstract}


\section{Introduction}
We frequently encounter scenarios where we need to estimate a finite collection of pmfs from their samples. We have a limited sample budget and we would like to adaptively choose which pmf to sample from so that we have uniformly good estimates of all the pmfs. The \emph{goodness} of our estimate can be evaluated using various distribution distance metrics. In this paper, we will focus on the mean squared error. A concrete application that fits this setting is that of estimating SARS-CoV-2 incidence rates in various geographic regions. Incidence rate, the fraction of \emph{currently} infected individuals, is a crucial metric for assessing the safety of a given region. Incidence rate cannot be estimated in an unbiased way simply based on the reported cases, it needs to be estimated by conducting randomized tests in each region of interest. Typically, such randomized surveys can be conducted by a central authority which has a limited testing budget. A natural question that arises in this scenario is that of allocating the available tests to various regions in an efficient manner. We may also be interested in other parameters such as SARS-CoV-2 seroprevalence which is an indicator of an individual's potential immunity against the virus, and other kinds of groups like ethnicity, age etc. in addition to the geographic location. Other applications of this setting include dynamics estimation in Markov Decision Processes (MDPs) and text compression which are discussed in greater detail in \cite{shekhar2020adaptive}.

When we are interested in minimizing the mean squared error for every pmf, an efficient sampling strategy would be to let the number of tests be proportional to the \emph{variance} of the associated pmf. However, the key problem is that we do not know the variance of the pmf. Therefore, we need to adapt our allocation as we learn the distribution over time. An upper confidence bound (UCB) based approached was proposed in \cite{shekhar2020adaptive}. The main idea in this approach is to first form an upper bound on the variance of each pmf. As a function of this upper bound and the number of samples acquired so far for each pmf, the next sample is chosen accordingly. It is shown that the number of samples obtained in this manner would approximately be the same as that of an \emph{oracle} that allocates samples using the knowledge of the variance of the underlying distributions. While this strategy satisfies many such interesting theoretical properties, we observed that it tends be more conservative than necessary in practice.

Our main goal in this paper is to devise a less conservative approach while ensuring the same performance guarantees as in \cite{shekhar2020adaptive}. Our contributions towards this goal are (i) we assume a Dirichlet prior on the pmfs and compute aforementioned upper confidence bounds using the posterior belief on the pmfs. This, combined with a stronger set of inequalities gives us noticeably tighter bounds on the variance of the pmfs; (ii) we analytically prove that this Bayesian approach for computing confidence bounds is no worse than the one in \cite{shekhar2020adaptive}; and (iii) we employ these methods for estimating SARS-CoV-2 seroprevalence in various categories using data obtained by a randomized survey \cite{sood2020seroprevalence} conducted in Los Angeles county. Our analysis of the survey data illustrates the importance of adaptive sampling in order to obtain better estimates. We would like to emphasize that the Bayesian approach used herein to compute upper confidence bounds and the associated proof methodologies can also be adapted to other UCB based approaches that are widely used in online learning \cite{efroni2020exploration,rosenberg2019online} and dynamics estimation \cite{tarbouriech2019active}.

\paragraph*{Related Work} Our work is closely related to \cite{shekhar2020adaptive}. We adopt their sampling strategy and many of the proof methodologies. The key difference between our work and \cite{shekhar2020adaptive} is our Bayesian approach for computing the upper confidence bounds on the variances. Our approach for computing these bounds has a significant impact on the performance in practice, especially in the non-asymptotic regime. Multiple distance metrics have been considered in \cite{shekhar2020adaptive}, whereas we restrict our attention to the mean squared error. It is however possible to extend our methodology to the other distance metrics as well. When the pmfs of interest have a binary support, which is the case in the problem of seroprevalence estimation, then the problem of estimating distributions reduces to the problem of estimating means of random variables. Substantial amount of work has been done in the area of estimating the means of random variables \cite{carpentier2011upper}. In \cite{carpentier2011upper}, two approaches for computing upper confidence bounds were proposed, one using the Chernoff-Hoeffding bound and the other using the empirical Bernstein bound. The Chernoff-Hoeffding bound approach is very close to that of \cite{shekhar2020adaptive}. We observe that the performance of our Bayesian approach is better than both the approaches in \cite{carpentier2011upper}. However, the regret bounds in \cite{carpentier2011upper} for the Bernstein approach are tighter than ours.

\paragraph*{Notation}\label{notation}
In general, subscripts denote time indices unless stated otherwise. For a sampling strategy $g$ and a collection of pmfs $p$, we use $\Py^g_p[\cdot]$ and $\E^g_p[\cdot]$ to indicate respectively that the probability and expectation depend on the choice of $g$, and that they are conditioned on the model $p$. We denote the indicator function associated with an event $\mathcal{E}$ with $\mathbbm{1}_{\mathcal{E}}$.

\section{Problem Formulation}\label{sec:probform}
We are interested in estimating $K$ probability mass functions (pmfs) over a finite space $\{1,\dots,L\}$ from their samples. Let the $k$-th pmf be denoted by $\s{p}{k} \doteq \left(\s{p}{k,1},\dots, \s{p}{k,L}\right)$ and let $p \doteq \left[\s{p}{1},\dots,\s{p}{K}\right]$. For convenience, we will refer to each of these pmfs as an \emph{arm}.

The total number of samples that can be obtained is $N$. At any given time, we can adaptively choose an arm to obtain a sample. The $n$-th arm sampled is denoted by $U_n$. Let the outcome of the $n$-th sample be denoted by $Y_n$. We have $Y_n \sim \s{p}{U_n}.$
Let the history of sampled arms and their corresponding outcomes be denoted by $I_n \doteq \{U_1,Y_1,\dots,U_{n-1},Y_{n-1}\}$. Let the strategy used to collect samples at time $n$ be $g_n$. In other words, $U_n = g_n(I_n)$. After collecting $n$ samples, we form an estimate $\s{\hp}{k}_n$ of each pmf $\s{p}{k}$, which in our case, is the empirical distribution. The mean squared error between the true distribution $\s{p}{k}$ and the estimate is denoted by 
\begin{align}
    \s{\mathcal{L}}{k}_N(\s{p}{k},g) = \E_p^g\left[\mathrm{MSE}\left(\s{p}{k},\s{\hp}{k}_{N}\right)\right].
\end{align}
\begin{remark}
Note that we can scale the loss $\s{\mathcal{L}}{k}$ by an appropriate weight $\s{w}{k}$. Our sampling strategy and proofs can be easily adapted to this weighted loss.
\end{remark}
As discussed earlier, our sampling strategy will involve the computation of upper confidence bounds on certain parameters. When Bayesian methods are used to derive upper confidence bounds, it is generally not possible to provide performance guarantees for every instance $p$ of the underlying pmfs. In such cases, the loss $\mathcal{L}$ is \emph{locally averaged} over a small region containing the true pmf $p$. We will therefore consider a loss that is locally averaged with respect to $p$ using a distribution $\varrho$. We refer to \cite{brown2001interval,bayarri2004interplay} for a detailed discussion on the interpretation of local averaging. 
\begin{remark}
Note that the distribution $\varrho$ used for local averaging is not the same as the prior on $p$. We will use a Bayesian interpretation of the pmfs $p$ only to compute tighter upper confidence bounds. However, our formulation and performance guarantees are frequentist in nature since the distribution $\varrho$ is supported on an arbitrarily small region around $p$.
\end{remark}

Let $\Delta$ denote the $L-1$ dimensional simplex and let $\varrho = \s{\varrho}{1}\times\dots\times\s{\varrho}{K}$ be a product distribution on $\Delta^K$.
Our goal is to minimize the following cost associated with a sampling strategy $g$ which is given by
\begin{align}
    \label{objective}J_N(g, \varrho) \doteq \max_k\left\{\E^{{\varrho}}\left[\s{\mathcal{L}}{k}_N(\s{\pi}{k},g)\right]\right\},
\end{align}
where ${\pi} \sim {\varrho}$. The distribution $\s{\varrho}{k}$ has the following form. For any Borel measurable set $\mathcal{A} \subseteq \Delta$,
\begin{align}
    \label{varrhodef}\s{\varrho}{k}(\mathcal{A}) = \mu\left(\mathcal{A}\cap\eta(\s{p}{k})\right)/\eta,
\end{align}
where $\mu$ is the uniform distribution over $\Delta$ and $\eta(\s{p}{k})$ is an $\ell_2$-ball around $\s{p}{k}$, such that $\mu\left(\eta(\s{p}{k})\right) = \eta$. Clearly, the cost in \eqref{objective} approaches the cost in \cite{shekhar2020adaptive,carpentier2011upper} when $\eta \rightarrow 0$ and thus, it is desirable to select small $\eta$. However, the performance guarantees we provide in this paper become weaker as $\eta$ becomes small. We will discuss the effects of the choice of $\eta$ in subsequent sections.

\paragraph*{The Oracle}
Let there be an oracle that knows the probability distributions $p$. Using this information, the oracle allocates a fixed number of samples for each arm. If the number of samples allocated to arm $k$ is $\s{T}{k}$, the mean-squared error associated with arm $k$'s estimate is given by
\begin{align}
    \E_p^g\left[\mathrm{MSE}\left(\s{p}{k},\s{\hp}{k}_{N}\right)\right] &= \frac{\sum_{l}\s{p}{k,l}(1-\s{p}{k,l})}{\s{T}{k}} \\
    &\doteq \frac{\s{c}{k}}{\s{T}{k}} \doteq \varphi(\s{c}{k},\s{T}{k}),
\end{align}
where $\varphi$ is referred to as the \emph{tracking function} and $\s{c}{k}$ is referred to as the \emph{tracking parameter} \cite{shekhar2020adaptive}. The oracle solves the following optimization problem to obtain an allocation.

\begin{align}
&\min_{\s{T}{1},\dots,\s{T}{K}} \;  \max_k \varphi(\s{c}{k},\s{T}{k})\tag{P1}\label{p}
&\mathrm{s.t.} \;  \sum_{k}\s{T}{k} = N.
\end{align}
While Problem \eqref{p} is a combinatorial optimization problem, an approximate solution can easily be obtained by making a convex relation. The non-integer solution and the optimal value for this problem can be expressed in closed form as
\begin{align}
    &\s{T}{k}_* = \frac{\s{c}{k}N}{\sum_{i}\s{c}{i}}; & \varphi^*(p,N) = \frac{\sum\s{c}{k}}{N}.
\end{align}
The regret with respect to the oracle that knows the underlying distributions is
\begin{align}
    \mathcal{R}_N(g,\varrho) \doteq J_N(g,\varrho) - \E^{{\varrho}}\left[\varphi^*(p,N)\right].
\end{align}

\section{Definitions and Framework}
We will first define some important quantities and compute an upper confidence bound on the tracking parameter $\s{c}{k}$. This bound will be used in our sampling strategy described in Section \ref{stratsec}.
\subsection{Prior Belief}
Let us assume that each of the pmfs $\s{p}{k}$ is independently drawn from a Dirichlet distribution \cite{kotz2004continuous} with parameters $\s{\alpha}{k}_1 \doteq \left(\s{\alpha}{k,1}_1,\dots,\s{\alpha}{k,L}_1\right)$.
In other words,
\begin{align}
    (\s{p}{1},\dots,\s{p}{K}) \sim \label{prior}\dir(\s{\alpha}{1}_1)\times\dots\times\dir(\s{\alpha}{K}_1).
\end{align}
This prior distribution\footnote{Note that this is distinct from the distribution $\varrho$ that was used for local averaging.} will be denoted by $\rho_1$. We will use the uniform prior which is equivalent to setting $\s{\alpha}{k,l}_1 = 1$ for every $k,l$. Thus, $\rho_1$ is the same as the distribution $\mu$ used in \eqref{varrhodef}.
\begin{remark}
Note that the prior $\rho_1$ is just something we assume. In a frequentist setting with local averaging, we are not provided with any prior information on the pmfs $p$. Since the distribution $\varrho$ is defined using the uniform distribution $\mu$ in \eqref{varrhodef}, the uniform prior $\rho_1 = \mu$ happens to be a convenient choice for computing our confidence bounds and subsequently analyzing them. In some cases however, we do have some prior information on $p$. A discussion on how to incorporate certain kinds of priors is in Appendix \ref{priorappend}.
\end{remark}

\subsection{Posterior Update}
When the prior on the pmfs is a factored Dirichlet distribution as in \eqref{prior}, the posterior distribution on the pmfs after collecting $n$ samples is $\prod_{k}\dir(\s{\alpha}{k}_{n+1})$ \cite{kotz2004continuous}, where
\begin{align}
    \label{postupdate}\s{\alpha}{k,l}_{n+1} &= \s{\alpha}{k,l}_{n} + \mathbbm{1}_{(U_n = k)}\mathbbm{1}_{(Y_n = l)}.
\end{align}
Further, since the pmf $\s{p}{k}$ is Dirichlet distributed, the probability $\s{p}{k,l}$ is distributed according to the Beta distribution. More precisely, conditioned on information $I_n$, we have
\begin{align}
    \label{betadis}\s{p}{k,l} \sim \mathrm{Beta}\left(\s{\alpha}{k,l}_n,\s{\alpha}{k,0}_n-\s{\alpha}{k,l}_n\right),
\end{align}
where $\s{\alpha}{k,0}_n \doteq \sum_{l=1}^L\s{\alpha}{k,l}_n$.
\subsection{Upper Confidence Bounds}
Let the number of times arm $k$ has been sampled until time $n$ be
\begin{align}
    \s{T}{k}_n = \sum_{\tau = 1}^{n-1}\s{U}{k}_\tau.
\end{align}
 Based on the posterior distribution of $\s{p}{k,l}$, we can compute a lower bound $\s{a}{k,l}_{n}$ and an upper bound $\s{b}{k,l}_{n}$ on $\s{p}{k,l}$ such that
\begin{align}
    \label{upp}\Py^{\rho_1}[\s{p}{k,l} > \s{b}{k,l}_{n} \mid I_n] &\leq \delta_n/2\\
    \label{low}\Py^{\rho_1}[\s{p}{k,l} < \s{a}{k,l}_{n} \mid I_n] &\leq \delta_n/2,
\end{align}
where $\delta_n = \frac{\delta}{KLn(1+\log N)}$ and $\delta = \eta N^{-5/2}$. Let us define an interval $\s{E}{k,l}_n \doteq [\s{a}{k,l}_{n},\s{b}{k,l}_{n}] \cap \s{E}{k,l}_{n-1}$. This ensures that the interval $\s{E}{k,l}_n$ is non-increasing in $n$. Note that we can compute this interval based on the inverse cdf of the posterior distribution, which is a Beta distribution.

For arm $k$, let $\s{q}{k}_n$ and $\s{u}{k}_n$ respectively be the optimal solution and the optimum value of the following quadratic program
\begin{align}
\label{ukupp}&\max_{q \in \Delta} \; \sum_l q_l(1-q_l) &\mathrm{s.t.} \; q_l \in \s{E}{k,l}_n\; \forall l.
\end{align}
Notice that since $\s{E}{k,l}_n$ is non-increasing in $n$, the upper bound $\s{u}{k}_n$ is also non-increasing in $n$. Let us define the event $\mathcal{E}$ as
\begin{align}
    \mathcal{E} \doteq \left\{\s{p}{k,l} \in \s{E}{k,l}_n\; \forall k,l,n\right\}.
\end{align}
Note that because of the way the upper bound $\s{u}{k,l}_n$ is defined in \eqref{ukupp}, we have $\s{p}{k,l}(1-\s{p}{k,l}) \leq \s{u}{k,l}_n$ under the event $\mathcal{E}$. Further, we have $\Py^{\rho_1}[\mathcal{E}] \geq 1-\delta$ because of the union bound $\sum_{n}\delta_n \leq \delta$.

\section{Strategies and Performance Bounds}\label{stratsec}

The main idea is simply to sample the arm that has maximum $\varphi(\s{u}{k}_n,\s{T}{k}_n)$ at time $n$. The precise strategy is stated as Algorithm \ref{algbasic}.


\begin{algorithm}

\DontPrintSemicolon
\caption{Bayesian UCB}
\label{algbasic}
\For{n = 1 \to N}{
\For{k = 1 \to K}{
Compute $\s{u}{k}_n$\;
}
Assign $U_n = \argmax_k \varphi(\s{u}{k}_n,\s{T}{k}_n)$\\
Sample arm $U_n$ to obtain $Y_n$ and update the posterior according to \eqref{postupdate}\;
}
\KwRet{$\s{\hp}{k}_{N}$}

\end{algorithm}

\subsection{Regret Bounds}
Using Algorithm \ref{algbasic}, we can achieve nearly the same mean squared error as the oracle would have achieved with the knowledge of underlying distributions. The gap between the oracle and Algorithm \ref{algbasic} is characterized by the regret bound below.


\begin{theorem}\label{mainthm}
Using the sampling strategy in Algorithm \ref{algbasic}, we have
\begin{align}
    \mathcal{R}_N(g,\varrho) \leq \mathcal{O} \left(\frac{\ln(\eta)}{N^{\frac{3}{2}}}\right).
\end{align}
\end{theorem}

Our methodology for deriving the regret bound in Theorem \ref{mainthm} is adapted from that in \cite{shekhar2020adaptive}. There are two key distinctions from the approach in \cite{shekhar2020adaptive}: (i) the gap $\s{e}{k}_n$ between the upper bound $\s{u}{k}_n$ and the true tracking parameter $\s{c}{k}$, and (ii) the use of local averaging with respect to the distribution $\varrho$. These distinctions arise due to the Bayesian nature of the confidence bounds in \eqref{ukupp}. In \cite{shekhar2020adaptive}, the gap $\s{e}{k}_n$ is bounded by a simple closed form expression (see Lemma 3 in \cite{shekhar2020adaptive}) which is then used to derive the regret bound. Characterizing this gap in our Bayesian setting is non-trivial due to the relatively complicated construction of the upper bound in \eqref{ukupp}. In the following two lemmas, we first derive a bound on $\s{e}{k}_n$.

\begin{lemma}\label{subgauss}
The length of the interval $\s{E}{k,l}_n$ satisfies
\begin{align}
    \mathrm{len}\left(\s{E}{k,l}_n\right) \leq \sqrt{\frac{2\ln{\frac{2}{\delta_n}}}{(\s{\alpha}{k,0}_n+1)}} \leq \sqrt{\frac{2\ln{\frac{2}{\delta_n}}}{(\s{T}{k}_n+1)}}.
\end{align}
\end{lemma}
\begin{proof}
See Appendix \ref{subgaussproof}.
\end{proof}
\begin{lemma}\label{l2bound}
Under event $\mathcal{E}$, we have
\begin{align}
    \s{e}{k}_n \doteq \s{u}{k}_n - \s{c}{k} \leq \sqrt{\frac{8\ln{\frac{2}{\delta_n}}}{(\s{T}{k}_n+1)}}.
\end{align}
\end{lemma}
\begin{proof}
See Appendix \ref{l2boundproof}.
\end{proof}
Note that the bound here is as tight as the one in Lemma 3 of \cite{shekhar2020adaptive}. These lemmas address the first aforementioned distinction from \cite{shekhar2020adaptive}. Using Lemma \ref{l2bound}, we can follow the steps of the proof of Theorem 1 in \cite{shekhar2020adaptive} to complete the proof of Theorem \ref{mainthm}. We refer the reader to Appendix \ref{mainthmproof} for the full proof with details related to local averaging.

\section{Numerical Results}
We consider an experimental setup with two arms ($K=2$) and a binary support ($L = 2$). The pmfs associated with arms 1 and 2 are $[0.99,0.01]$ and $[0.7,0.3]$. The total number of samples $N = 2500$ and the value of $\eta = 1/N$. Figure \ref{regfig} depicts the regret $\mathcal{R}_N$ for the sampling strategy in \cite{shekhar2020adaptive} and our sampling strategy in Algorithm \ref{algbasic}. We will refer to the former strategy as UCB and the latter as Bayesian UCB. There is a clear improvement in performance in terms of regret. We also plot the average number of samples acquired by both strategies in Figure \ref{samplefig}. We observe that the number of samples acquired by our Bayesian UCB approach is closer to the oracle allocation. 

\begin{figure}
     \centering
     \subfigure[][b]{\includegraphics[width=0.32\textwidth]{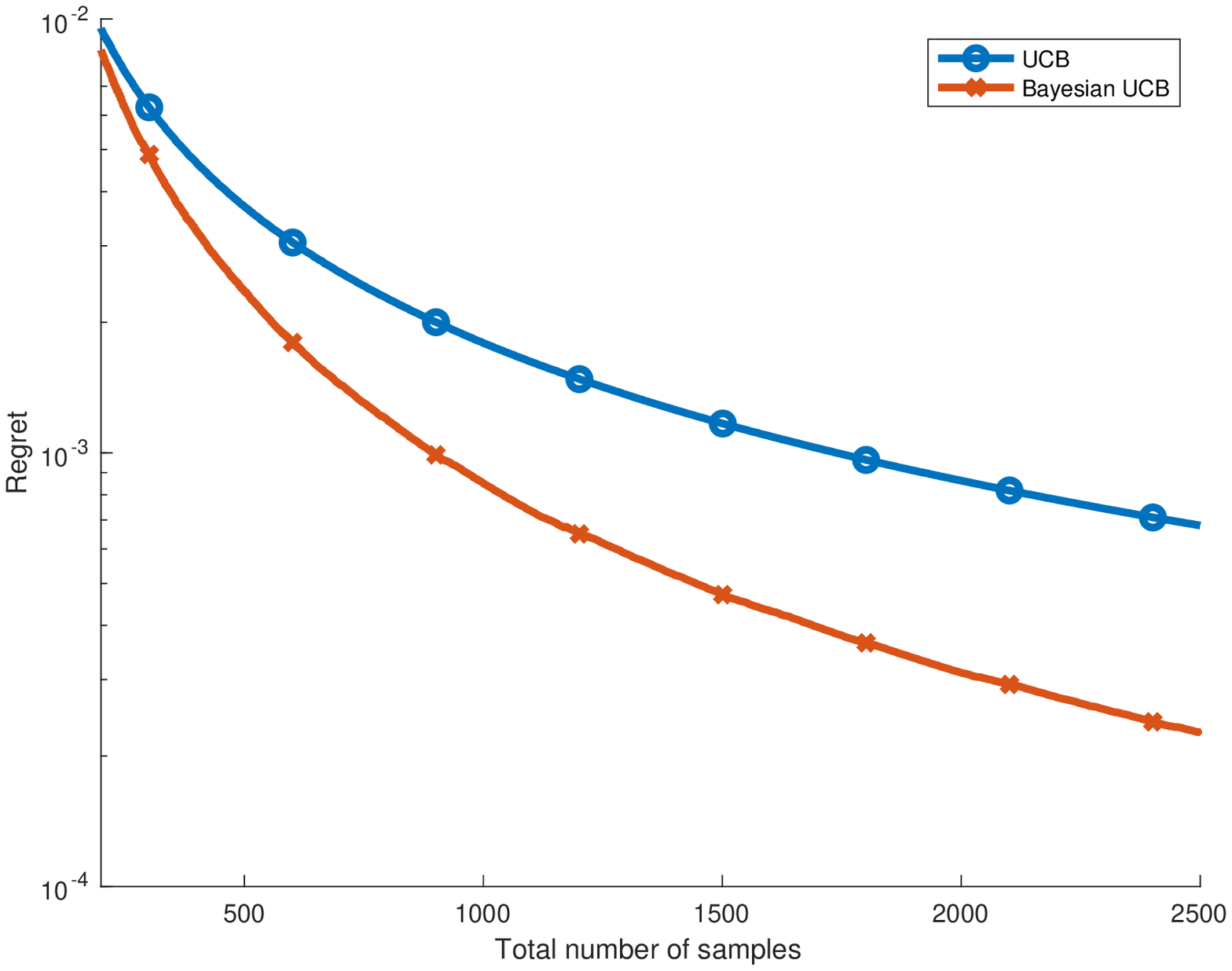}\label{regfig}}
     \subfigure[][b]{\includegraphics[width=0.32\textwidth]{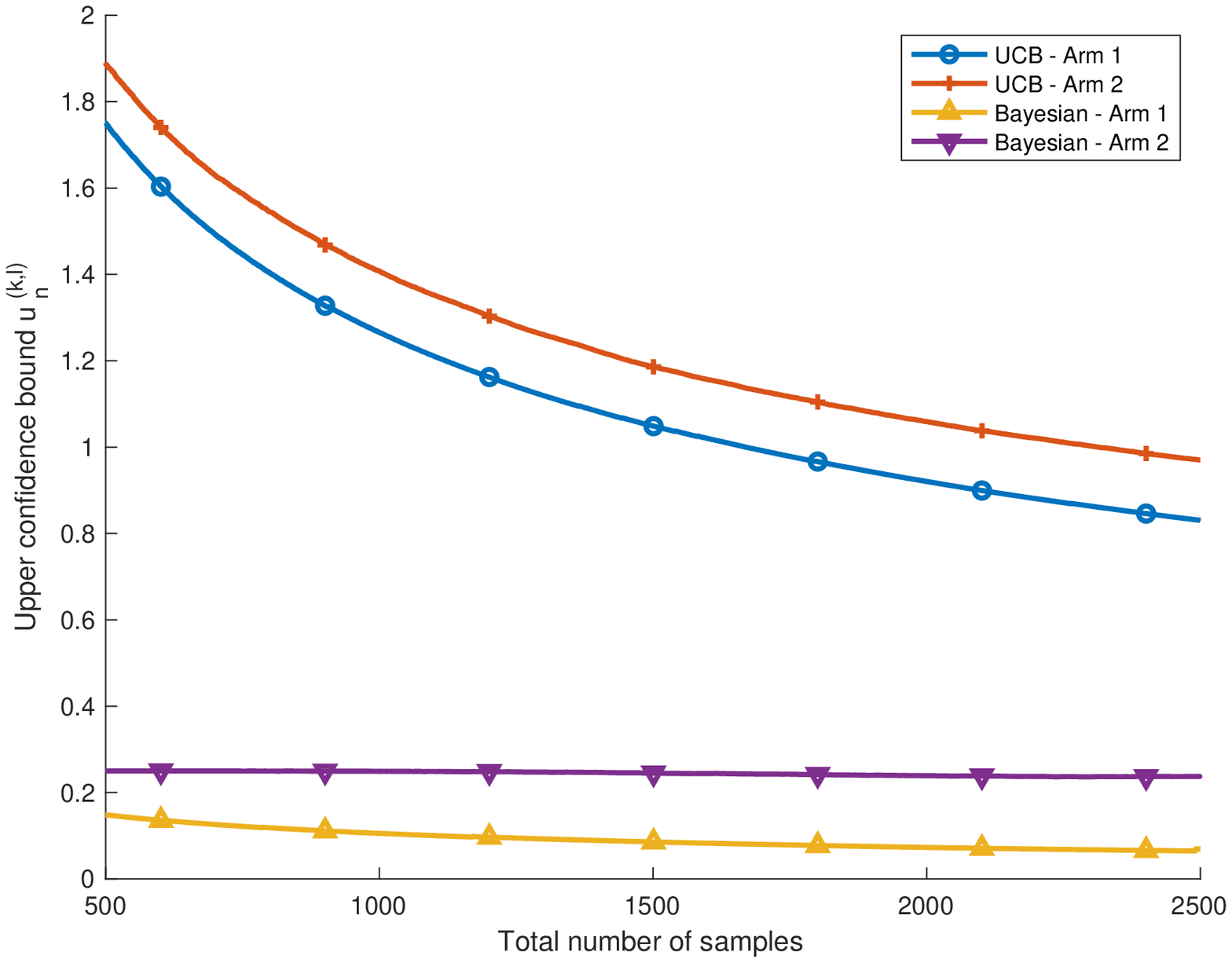}\label{uppfig}}
     \subfigure[][b]{\includegraphics[width=0.32\textwidth]{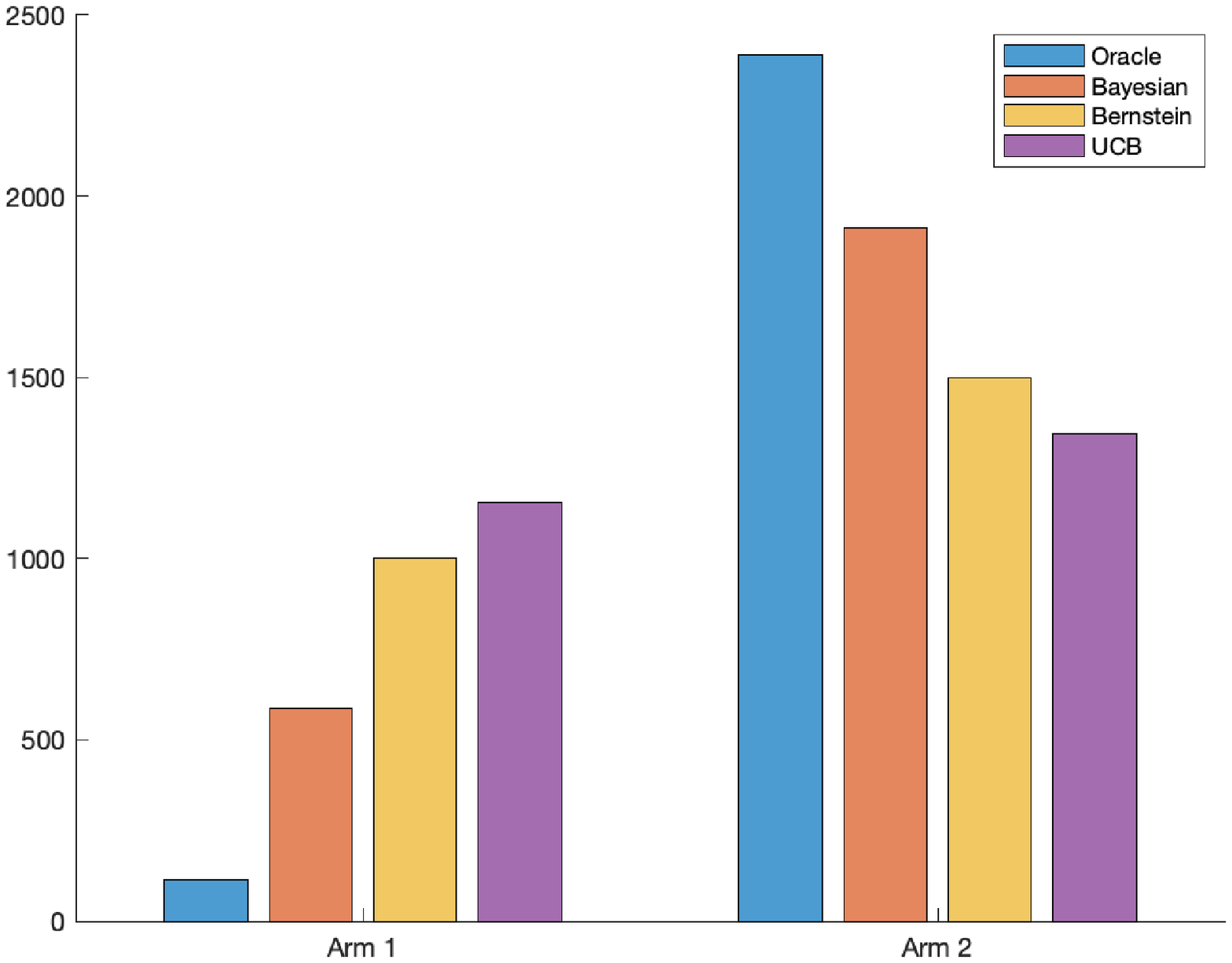}\label{samplefig}}
        \caption{The plot on the left depicts the regret $\mathcal{R}_N$ for the strategy in Algorithm \ref{algbasic} vs. the strategy in \cite{shekhar2020adaptive}. The plot in the middle depicts the upper confidence bounds $\s{u}{k}_n$. The plot on the right depicts the expected number of samples obtained and includes the empirical Bernstein approach in \cite{carpentier2011upper} as well.}
        \label{fig:three graphs}
\end{figure}

Recall that the main difference between the UCB and the Bayesian UCB strategies is the way the upper bound $\s{u}{k,l}_n$ is computed. In both these approaches, we observe that for every component, the value of the tracking function $\s{u}{k}_{N+1}/\s{T}{k}_{N+1}$ is nearly the same for all arms $k$.
\begin{proposition}\label{propobs}
For every arm $k$, we have
\begin{align}
    \frac{\sum_k \s{u}{k}_{N+1}}{N}\times\frac{1}{f_N} \leq \frac{\s{u}{k}_{N+1}}{\s{T}{k}_{N+1}} &\leq \frac{\sum_k \s{u}{k}_{N+1}}{N}\times f_N,
\end{align}
where $f_N = 1+o(1)$. Therefore, $\s{T}{k}_{N+1} \propto \s{u}{k}_{N+1}$, approximately.
\end{proposition}
In Figure \ref{uppfig}, we plot a typical realization of the upper bounds $\s{u}{k}_n$ as a function of time $n$. Notice that the bounds computed in our approach are significantly tighter. This tightness can be attributed to two factors: (i) confidence intervals on the probabilities $\s{p}{k,l}$ computed using the inverse cdf of the posterior distribution, and (ii) the quadratic program for computing $\s{u}{k}_n$ in \eqref{ukupp}. The bounds computed using the method in \cite{shekhar2020adaptive} do capture the \emph{difference} between the tracking parameters for arms 1 and 2. However, according to Proposition \ref{propobs}, it is the \emph{ratio} between the upper bounds that matters. The ratio between the upper bounds in \cite{shekhar2020adaptive} is close to 1 even for $N=2500$. While this ratio will approach $\s{c}{1}/\s{c}{2}$ eventually, it is quite slow in practice.
While our Bayesian approach provides tighter bounds, there is a computational cost associated with it, especially because of the diminishing probability $\delta_n$. For the binary support case, there are simpler closed form approximations \cite{bayarri2004interplay,brown2001interval} that can be used to compute $\s{a}{k,l}_n$ and $\s{b}{k,l}_n$.

We propose to use our sampling strategy for tracking SARS-CoV-2 seroprevalence in populations classified on the basis of location, ethnicity, age etc. Each category of interest can be modeled as an arm with a Bernoulli pmf, where the probability $\s{p}{k,2}$ represents the positivity associated with category $k$. We use data from a randomized seroprevalence survey conducted in Los Angeles county \cite{sood2020seroprevalence}. The plots in Figures \ref{racefig} and \ref{regionfig} depict the gap between optimal sample allocation and the random allocation strategy used by the survey. This gap illustrates the importance of adaptive randomized sampling for obtaining uniform estimates of seroprevalence. Without careful sampling, we can have bad estimates for certain groups that may be particularly vulnerable.
\begin{figure}
     \centering
     \subfigure[][t]{\includegraphics[width=0.49\textwidth]{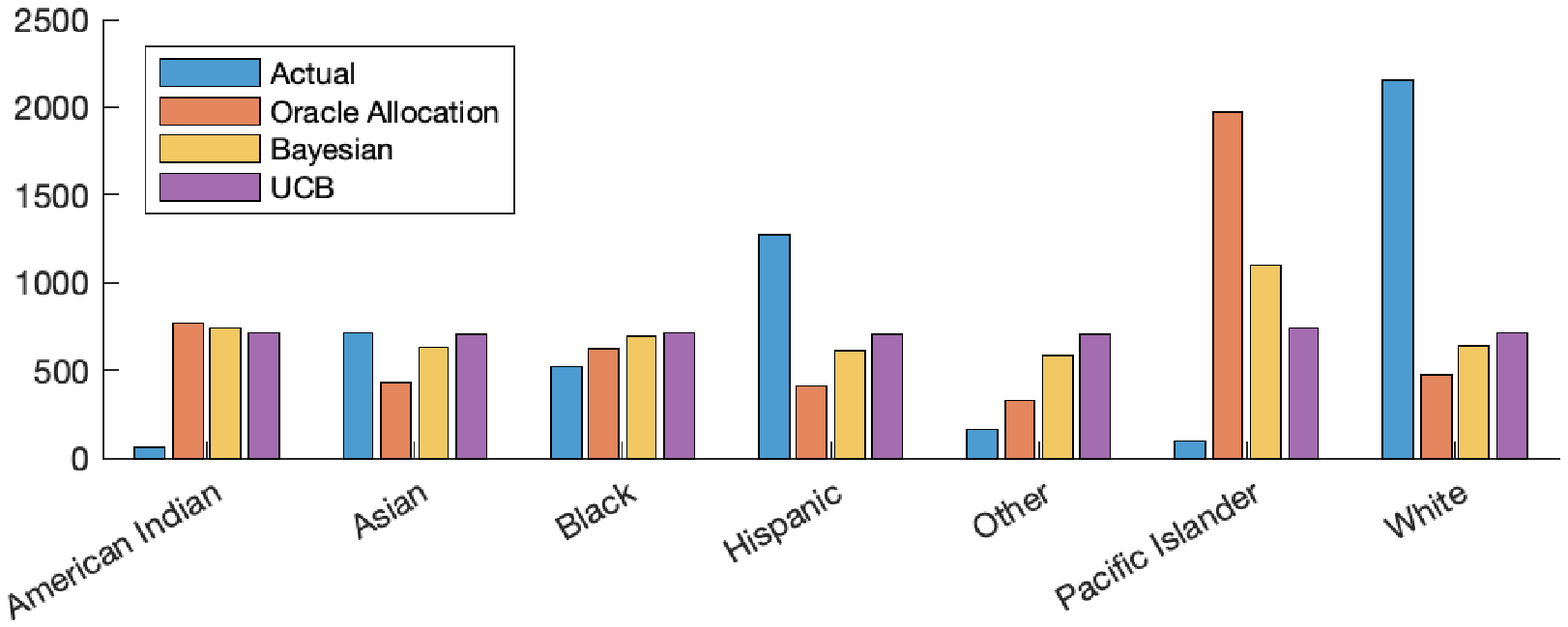}\label{racefig}}
     \subfigure[][t]{\includegraphics[width=0.49\textwidth]{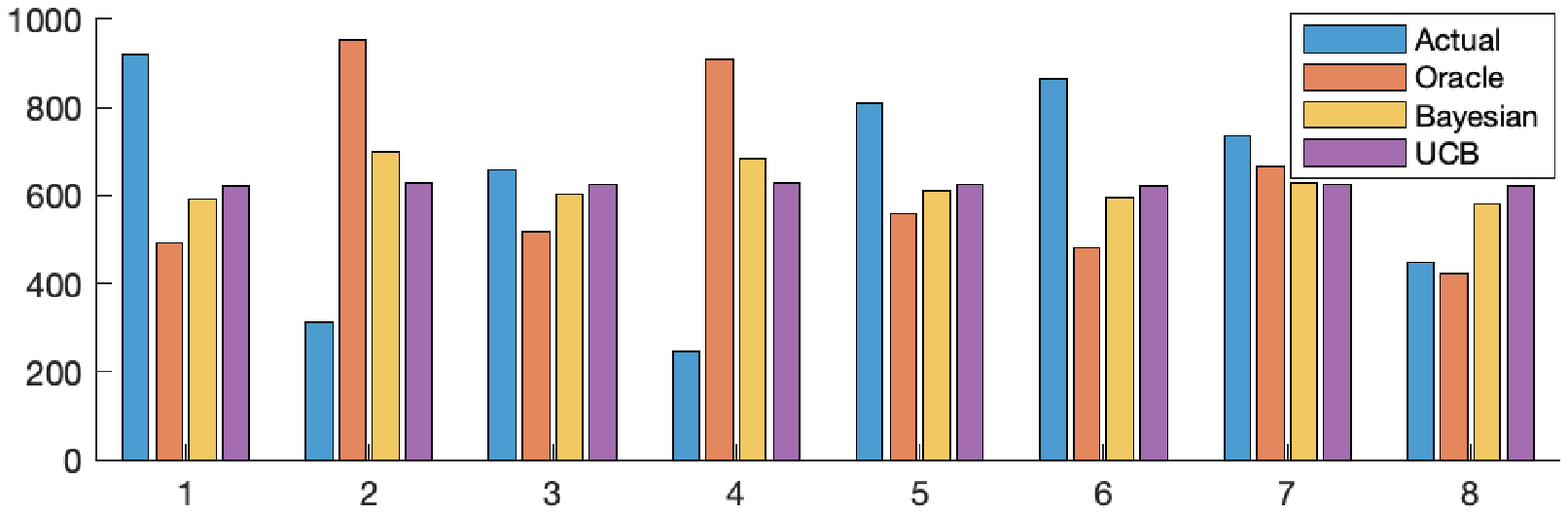}\label{regionfig}}
        \caption{Population classified on the basis of ethnicity and location, respectively in plots \ref{racefig} and \ref{regionfig}. The plots depict the number of sampled obtained in each category under each strategy. Here, \emph{Actual} refers to the number of samples obtained by the survey \cite{sood2020seroprevalence}}
        \label{fig:three graphs}
\end{figure}

\begin{remark}
For estimating SARS-CoV-2 seroprevalence, there are some practical aspects that require a slightly different formulation and sampling strategy. We provide a heuristic approach for addressing these issues and the analysis of this approach is a problem for future work.
\end{remark}
\section{Conclusions and Future Work}
For the problem of uniformly estimating pmfs in a mean squared sense, a UCB sampling approach is proposed. Unlike the state-of-the-art approach in \cite{shekhar2020adaptive}, a Dirichlet prior is assumed on the pmfs and the upper confidence bounds are computed based on the posterior. It is analytically shown that the performance of this Bayesian approach can be now worse than the approach in \cite{shekhar2020adaptive}, and has noticeably better performance in practice. We discuss the potential application of the sampling approach for SARS-CoV-2 seroprevalence and incidence rate estimation.

There are two potential directions for future work. One is to incorporate temporal variation in the underlying pmfs. This is particularly important in studying SARS-CoV-2 incidence rates. The other is to incorporate a prior on the collection pmfs so that the correlation between them can be exploited to improve our estimates and the sampling strategy. For both these directions, a Bayesian model can be more helpful.





\appendix

\section{Proof of Lemma \ref{subgauss}}\label{subgaussproof}
A Beta distribution with parameters $\alpha,\beta$ is sub-Gaussian with parameter $\frac{1}{4(\alpha+\beta+1)}$ \cite{marchal2017sub}. And because of \eqref{betadis} we can conclude using the Chernoff bound \cite{ross2014introduction} that
\begin{align}
    \Py^{\rho_n}\left[\s{p}{k,l} > \frac{\s{\alpha}{k,l}_n}{\s{\alpha}{k,0}_n} + \sqrt{\frac{\ln\frac{2}{\delta_n}}{2(\s{\alpha}{k,0}_n+1)}}\right] &\leq \delta_n\\
    \Py^{\rho_n}\left[\s{p}{k,l} < \frac{\s{\alpha}{k,l}_n}{\s{\alpha}{k,0}_n} - \sqrt{\frac{\ln\frac{2}{\delta_n}}{2(\s{\alpha}{k,0}_n+1)}}\right] &\leq \delta_n.
\end{align}
Since $\s{a}{k,l}_n$ and $\s{b}{k,l}_n$ are selected using the inverse cdf of the posterior distribution, we have
\begin{align}
    \s{b}{k,l}_n - \s{a}{k,l}_n \leq 2\sqrt{\frac{\ln\frac{2}{\delta_n}}{2(\s{\alpha}{k,0}_n+1)}}.
\end{align}
The second inequality in the lemma simply follows from the fact the $\s{\alpha}{k,0}_n \geq \s{T}{k}_n$.

\section{Proof of Lemma \ref{l2bound}}\label{l2boundproof}
Under event $\mathcal{E}$, we have
\begin{align}
    &|\s{u}{k}_n - \s{c}{k}|\\
    &= \left|\sum_{l} \left((\s{p}{k,l})^2 -(\s{q}{k,l}_n)^2  - \s{p}{k,l} + \s{q}{k,l}_n\right) \right|\\
    &\stackrel{a}{\leq} \sum_{l}\left(\left| (\s{p}{k,l})^2 -(\s{q}{k,l}_n)^2\right|\right)\\
    &=\sum_{l}\left| \s{p}{k,l} -\s{q}{k,l}_n  \right|\left(\s{p}{k,l} +\s{q}{k,l}_n\right)\\
    &\leq \sum_{l}\mathrm{len}\left( \s{E}{k,l}_n  \right)\left(\s{p}{k,l} +\s{q}{k,l}_n \right)\\
    &\stackrel{b}{\leq} 2\sqrt{\frac{\ln{\frac{2}{\delta_n}}}{2(\s{T}{k}_n+1)}}\sum_{l}\left(\s{p}{k,l} +\s{q}{k,l}_n \right)\\
    &= 4\sqrt{\frac{\ln{\frac{2}{\delta_n}}}{2(\s{T}{k}_n+1)}}.
\end{align}
Here, inequality in $(a)$ is because of the triangle inequality and the inequality in $(b)$ is due to Lemma \ref{subgauss}.

\section{Proof of Theorem \ref{mainthm}}\label{mainthmproof}
We have
\begin{align}
    \delta \geq \Py^{\rho_1}[\mathcal{E}^c] &= \E^{\rho_1}\left[\delta_p\right]\\
    &\geq \E^{\rho_1}\left[\delta_p\mathbbm{1}_{\eta(p)}(p)\right]\\
    &= \E^{\varrho}\left[\delta_p\right]\eta,\label{measurechange}
\end{align}
where $\delta_p \doteq \Py^g[\mathcal{E}^c\mid p]$. Note that the last equality in the display above follows from a simple change of measure argument.
Since $\delta_p$ is the probability of the unfavorable event $\mathcal{E}^c$ conditioned on the model $p$, we can obtain a regret bound in terms of $\delta_p$ without local averaging using the same arguments as in the proofs of Lemma 1 and Theorem 1 in \cite[Appendix D.3]{shekhar2020adaptive}. Thus, we have
\begin{align}
    \s{\mathcal{L}}{k}_N(&\s{p}{k},g)-\varphi^*(p,N)\\
    &\leq \frac{(K+5)LM}{(\lambda_{\mathrm{min}})^2N^{3/2}} + \frac{\delta_p}{(\lambda_{\mathrm{min}})^2N}\left(1+6M\sqrt{N}\right)+ \frac{6(K-1)M^2}{(\lambda_{\mathrm{min}})^3N^2} + L\delta_p,
\end{align}
where
\begin{align}
    M &\doteq \frac{\lambda_{\mathrm{max}}\sqrt{8\log(2/\delta_N)}}{\lambda_{\mathrm{min}}C} = \mathcal{O}\left(\sqrt{\log(\eta N)}\right)\\
    \lambda_{\mathrm{max}} &\doteq \min_{k,\pi\in\eta(
    p)}\frac{\s{c}{k}}{\sum_i\s{c}{i}};
    \quad\lambda_{\mathrm{max}}\doteq\max_{k,\pi\in\eta(
    p)}\frac{\s{c}{k}}{\sum_i\s{c}{i}}\\
    C &\doteq \sum_i\s{c}{i},
\end{align}
and $\s{c}{i} = \sum_{l}\s{\pi}{i,l}(1-\s{\pi}{i,l})$. Using \eqref{measurechange}, we can then locally average the regret as
\begin{align}
    &\E^{{\varrho}}\left[\s{\mathcal{L}}{k}_N(\s{\pi}{k},g)-\varphi^*(p,N)\right]\\
    &\leq \frac{(K+5)LM}{(\lambda_{\mathrm{min}})^2N^{3/2}} + \frac{\delta}{(\eta\lambda_{\mathrm{min}})^2N}\left(1+6M\sqrt{N}\right)+ \frac{6(K-1)M^2}{(\lambda_{\mathrm{min}})^3N^2} + L\frac{\delta}{\eta}\\
    &= \mathcal{O}(\log(\eta) N^{-3/2}).
\end{align}



\section{Incorporating Priors}\label{priorappend}
In this section, we will consider a scenario where we are provided with a prior $\rho_1$ on $p$. We will first modify the distribution $\varrho$ used for local averaging as
\begin{align}
    \label{varrhodef}{\varrho}(\mathcal{A}) = \rho_1\left(\mathcal{A}\cap\eta({p})\right)/\eta,
\end{align}
for any Borel measurable set $\mathcal{A} \subseteq \Delta^K$. Here, $\eta(p)$ is a ball around $p$ with $\rho_1(\eta(p)) = \eta$.

\paragraph{Dirichlet priors}When the prior $\rho_1$ is any factored Dirichlet distribution (\emph{i.e}. not necessarily uniform) of the form in \eqref{prior}, one can easily extend all the results with the modified $\varrho$.
\paragraph{Intervals for Bernoulli parameters}Another scenario of interest is when $L=2$ and we are given that the Bernoulli parameter associated with $\s{p}{k}$ lies in some interval $[\s{\gamma}{k}_l,\s{\gamma}{k}_u]$. We can then assume that the Bernoulli parameter $\s{p}{k}$ is distributed uniformly and independently over the interval $[\s{\gamma}{k}_l,\s{\gamma}{k}_u]$. Since this prior does not correspond to a Beta distribution, the posterior belief may not be a Beta distribution. However, we still can easily find the cdf of the posteriors in this case because the posterior on the parameters is just going to be a truncated Beta distribution. Let $\s{F}{k}_n$ be the cdf of the posterior at time $n$ with a uniform prior and let $\s{\tilde{F}}{k}_n$ be the posterior with the truncated prior. Then one can show that
\begin{align}
    \s{\tilde{F}}{k}_n(x) = \frac{\s{F}{k}_n(x)-\s{F}{k}_n(\s{\gamma}{k}_l)}{\s{F}{k}_n(\s{\gamma}{k}_u)-\s{F}{k}_n(\s{\gamma}{k}_l)}.
\end{align}
The confidence intervals in \eqref{upp} and \eqref{low} can then be computed using the inverse of $\s{\tilde{F}}{k}_n$.

\section{Practical Considerations: Overall Estimate and Batch Sampling}
In seroprevalence estimation, we generally allocate samples in batches of size $B$. Also, we are generally interested in estimating the positivity of each category \emph{as well as} the positivity in the overall population. Let the fraction of individuals of category $k$ in the overall population be $\s{w}{k}$. Then the overall positivity $r$ and its estimate $\hat{r}$ are given by
\begin{align}
    r &= \sum_{k}\s{w}{k}\s{p}{k}\\
    \hat{r}_N &= \sum_{k}\s{w}{k}\s{\hat{p}}{k}_N.
\end{align}
The mean squared error between $r$ and $\hat{r}_N$ is given by
\begin{align}
    \mathrm{MSE}(r,\hat{r}_N) = \sum_k\frac{(\s{w}{k})^2\s{c}{k}}{\s{T}{k}}.
\end{align}
If the mean squared error associated with $r$ is not considered, then it may so happen that a tiny group (small $w_k$) with high positivity will be allocated too many samples. The contribution of this small group to the overall estimate would be small and thus, allocating too many samples to it could compromise the quality of the overall estimate $r$. Therefore, we need to determine an allocation that accounts for the quality of the overall estimate $r$ as well.

A suitable way to formalize this notion is to pose the following constraints on the oracle allocation $\s{T}{k}$
\begin{align}
    \tag{C1}\label{c1}&\frac{\s{c}{k}}{\s{T}{k}}\leq \s{\theta}{k},\; k =1,\dots,K\\
    & \sum_k \frac{(\s{w}{k})^2\s{c}{k}}{\s{T}{k}} \leq \s{\theta}{0}\\
    & \sum_{k}\s{T}{k} \leq N\\
    & \s{T}{k} \geq 0,\; k =1,\dots,K,
\end{align}
where $\s{\theta}{k}$, $k=0,\dots,K$ are predetermined constants. A solution to the above set of constraints can be obtained by solving the following optimization problem
\begin{align}
    \tag{C2}\label{c2}\min_{\s{T}{1},\dots,\s{T}{K}} \quad & \max\left\{\max_k\left\{\frac{\s{c}{k}}{\s{\theta}{k}\s{T}{k}}\right\},\sum_k \frac{(\s{w}{k})^2\s{c}{k}}{\s{\theta}{0}\s{T}{k}} \right\}\\
\mathrm{s.t.} \quad & \sum_{k}\s{T}{k} \leq N\\
    & \s{T}{k} \geq 0,\; k =1,\dots,K.
\end{align}
The Problem \eqref{c1} is feasible if and only if the optimum value of the optimization problem above is less than or equal to 1. In that case, the solution to Problem \eqref{c2} is a solution to Problem \eqref{c1}. Notice that Problem \eqref{c1} is very similar to Problem \eqref{p} except for the additional mean squared error term associated with the overall estimate $r$. Because of this distinction, it is not clear whether one can view Problem \eqref{c2} as a particular instance of Problem \eqref{p} with appropriate modifications and simply apply the adaptive sampling strategy in Section \ref{stratsec} to the modified problem. Nonetheless, we provide a similar heuristic sampling approach that tracks the oracle quite well (See Figure \ref{fig:jointplot}). At each time time $n$, the heuristic is to allocate samples within each batch according to the solution of the following optimization problem
\begin{align}
    \tag{C3}\label{c3}\min_{\s{\tau}{1},\dots,\s{\tau}{K},\lambda} \quad & \lambda\\
\mathrm{s.t.} \quad & \frac{u_n^{(k)}}{T_{n}^{(k)}+\s{\tau}{k}} \leq \s{\theta}{k}\lambda\\
& \sum_k \frac{(\s{w}{k})^2u_n^{(k)}}{T_{n}^{(k)}+\s{\tau}{k}} \leq \s{\theta}{0}\lambda\\
&\sum_{k}\s{\tau}{k} = B\\
& \s{\tau}{k} \geq 0 \; \forall k.
\end{align}
\begin{figure}
    \centering
    \includegraphics[width = 0.5\textwidth]{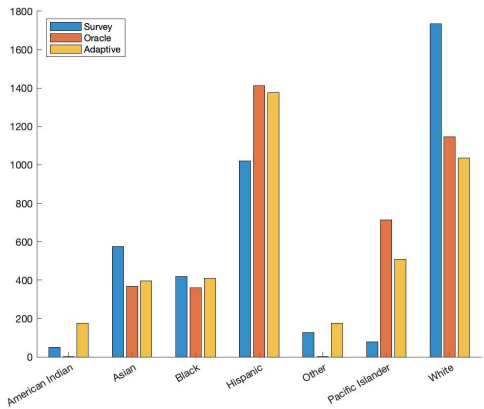}
    \caption{This plot represents the number of samples collected by the seroprevalence survey in \cite{sood2020seroprevalence}, the oracle allocation in Problem \eqref{c1} with appropriate constants $\s{\theta}{k}$, and the allocation by the heuristic (denoted by Adaptive). Notice that the heuristic tracks the oracle closely. This plot tells us that in order get a good overall estimate, we should allocate fewer samples to an underrepresented group like the Pacific Islanders ($\s{w}{k} \approx 0.003$) than the number suggested by the oracle \eqref{p} (See \ref{racefig}). However, it also tells us that we can allocate substantially more samples than the number in the survey ($\approx \s{w}{k}N$) to this group for a better estimate of their positivity without comprising the quality of the overall estimate.}
    \label{fig:jointplot}
\end{figure}

\acks{This research was supported, in part, by National Science Foundation under Grant NSF CCF-1817200, CCF-1718560, CPS-1446901, Grant ONR N00014-15-1-2550, and Grant ARO {W911NF1910269}.}
\bibliography{refs.bib}

\end{document}